\documentclass[letterpaper, 10 pt, conference]{ieeeconf}
\IEEEoverridecommandlockouts
\usepackage{amsthm}
\usepackage{cite}
\usepackage{amsmath,amssymb,amsfonts} 
\usepackage{algorithmic}
\usepackage{graphicx}
\usepackage{graphbox}
\usepackage{textcomp}
\usepackage[dvipsnames]{xcolor}\usepackage{float}
\usepackage[caption=false,font=scriptsize]{subfig}
\usepackage{tikz}
\usetikzlibrary{positioning}
\usepackage{calc}
\usepackage{comment}
\usepackage{tabularray}
\usepackage{booktabs}
\usepackage{dsfont}
\makeatletter
\let\NAT@parse\undefined
\makeatother
\usepackage{hyperref}
\usepackage{cleveref}
\usepackage{adjustbox}

\usepackage[normalem]{ulem}

\def\BibTeX{{\rm B\kern-.05em{\sc i\kern-.025em b}\kern-.08em
    T\kern-.1667em\lower.7ex\hbox{E}\kern-.125emX}}
\theoremstyle{remark}

\newtheorem{lemma}{Lemma}

\newtheorem{proposition}{Proposition}

\newcommand{\CC}[1]{{\color{ForestGreen} (CC: #1)}}
\newcommand{\IX}[1]{{\color{Melon} (IX: #1)}}

\newcommand{\edit}[1]{{\color{blue}#1}}

\newcommand{\squeezeup}{\vspace{-2.5mm}}

\begin{document}

\title{\LARGE \bf Spiking Nonlinear Opinion Dynamics (S-NOD)\\for Agile Decision-Making}

\author{Charlotte Cathcart$^1$, Ian Xul Belaustegui$^1$, Alessio Franci$^{2\dagger}$, and Naomi Ehrich Leonard$^{1\dagger}$
\thanks{$^1$Mechanical and Aerospace Engineering, Princeton University, Princeton, NJ, USA. 
{\tt {\{cathcart, ianxul, naomi\}@princeton.edu}}}
\thanks{$^2$Electrical Engineering and Computer Science, University of Liège, Liège, Belgium and WEL Research Inst., Wavre, Belgium. 
{\tt afranci@uliege.be}
}
\thanks{This research has been supported by ONR grant N00014-19-1-2556.}
\thanks{$^\dagger$These authors contributed equally to this work.}
}
\maketitle
\begin{abstract}
We present, analyze, and illustrate a first-of-its-kind model of two-dimensional excitable (spiking) dynamics for decision-making over two options. The model, Spiking Nonlinear Opinion Dynamics (S-NOD), provides superior agility, characterized by fast, flexible, and adaptive response to rapid and unpredictable changes in context, environment, or information received about available options.  
 S-NOD derives through the introduction of a single extra term to the previously presented Nonlinear Opinion Dynamics (NOD) for fast and flexible multi-agent decision-making behavior. The extra term is inspired by the fast-positive, slow-negative mixed-feedback structure of excitable systems. The agile behaviors brought about by the new excitable nature of decision-making driven by S-NOD are analyzed in a general setting and  illustrated in an application to multi-robot navigation around human movers.
\end{abstract}

\section{Introduction}\label{sec: newintro}

The fast, flexible, and adaptive behavior observed in biology owes much to the excitable (spiking) nature of cellular signaling~\cite{dayan2005theoretical,Balazsi2011,suel2006excitable,fromm2007electrical}. Models of excitability 
represent the analog molecular and/or biophysical processes that produce spikes in response to stimuli. These models inherit  
the adaptive behavior of analog systems and the reliability of digital systems, foundational to 
spiking control systems~\cite{sepulchre2022spiking} and neuromorphic engineering~\cite{bartolozzi2022embodied}. However, existing models describe single-input/single-output~\cite{franci2019sensitivity} spike-based signal processing. 
This is spiking activity that can only encode single-option decisions: to spike or not spike, as determined by the input signal pushing the system toward its excitability threshold~\cite{ref:Sepulchre_ExcitableBehaviorsChapter2018}. This limits the use of these models in studying and designing spiking decision-making over \textit{multiple} options as observed, e.g., in cortical column visual orientation selectivity~\cite{Priebe2016} or in sensorimotor decision-making~\cite{gallivan2018decision}.

We present a generalized model of excitable (spiking) dynamics that allows for fast, flexible, and adaptive decision-making over multiple options. In this paper we focus on two-option spiking in a two-dimensional, two-timescale model and we use ``agile'' to mean ``fast, flexible, and adaptive.'' To the best of our knowledge, this is the first such model to generalize spiking to more than one option, i.e., spiking that can occur in any of the multiple directions corresponding to the multiple options available to the excitable decision-maker. We call our model \textit{Spiking Nonlinear Opinion Dynamics} (S-NOD) as it derives from the introduction of an extra term in the Nonlinear Opinion Dynamics (NOD) model of~\cite{ref:Bizyaeva_NODTunable2023,ref:Leonard_ARpaper2024} that makes NOD excitable, i.e., spiking. 

NOD model the time evolution of opinions of a group of agents engaged in a collective decision-making process over a set of options. The derivation of NOD was tailored to model and study the principles of 
fast and flexible decision-making in biological collectives~\cite{VS-NEL:16c,gray2018multiagent}
and to use these principles to design 
fast and flexible decision-making in built collectives~\cite{ref:Bizyaeva_NODTunable2023,ref:franci2021analysis,bizyaeva2023multi}. 
NOD exhibits a mixed-feedback~\cite{sepulchre2019control, ref:Sepulchre_ExcitableBehaviorsChapter2018} structure: 
opinion formation arises from the balance of a negative feedback loop that regulates agent opinions to a neutral solution and 
positive feedback loops (at the single-agent and network levels) that destabilize the neutral solution and trigger nonlinear opinion formation.
Decision-making driven by NOD  \cite{ref:Leonard_ARpaper2024} is fast 
because it 
can diverge quickly from indecision even in the absence of informative inputs about the options. 
It is flexible because 
the sensitivity of opinion formation to informative inputs 
is tunable. 
Both speed and flexibility 
are determined by 
a tunable threshold for opinion formation where 
negative and positive feedback are perfectly balanced and the dynamics become singular.

\begin{figure}[t]
	\centering
	\subfloat[]{\includegraphics[width=0.5\columnwidth]{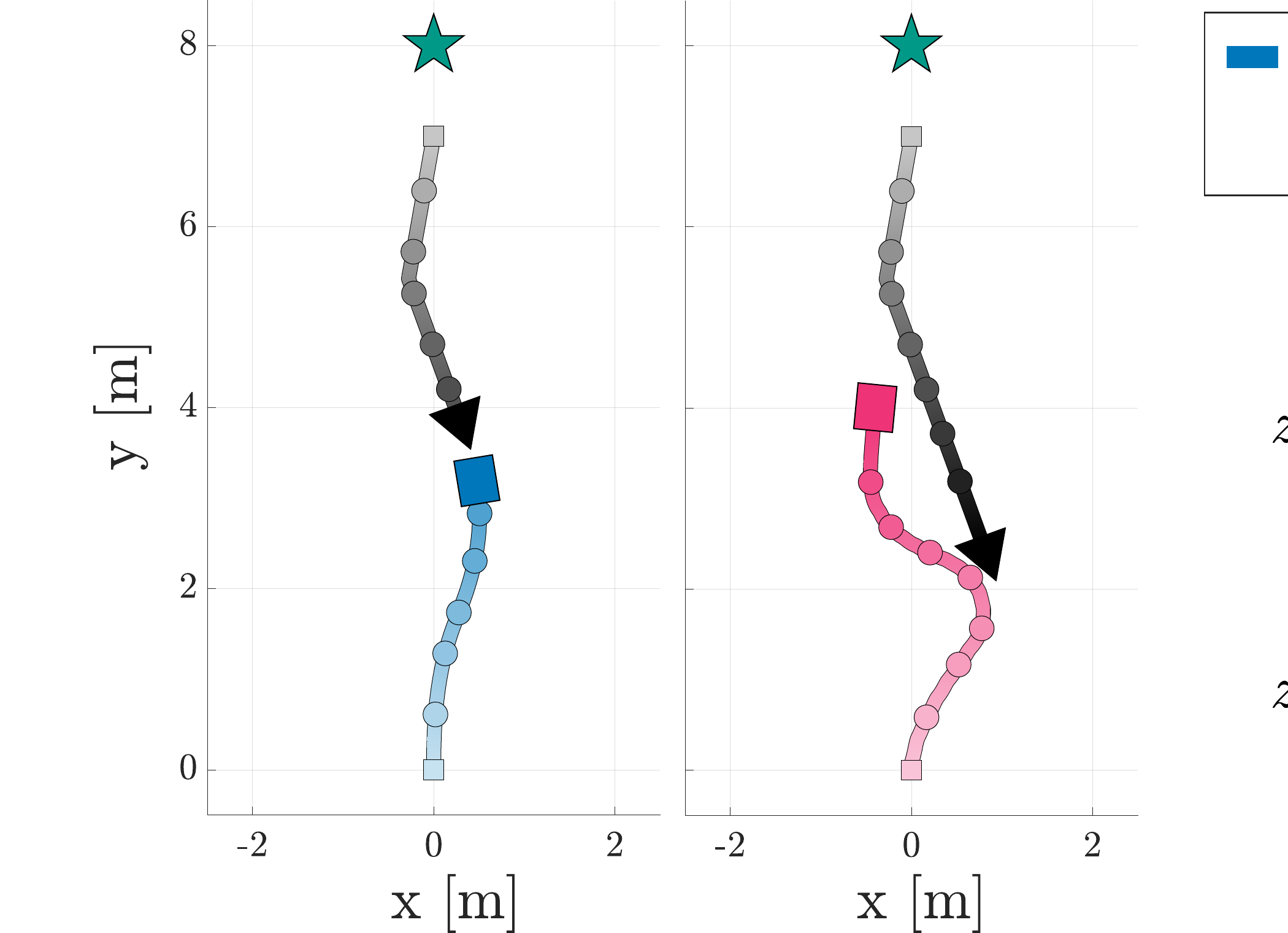}}\label{subfig:intro_robot_traj}\hspace{-0.13cm}
  \subfloat[]{\includegraphics[width=0.5\columnwidth]{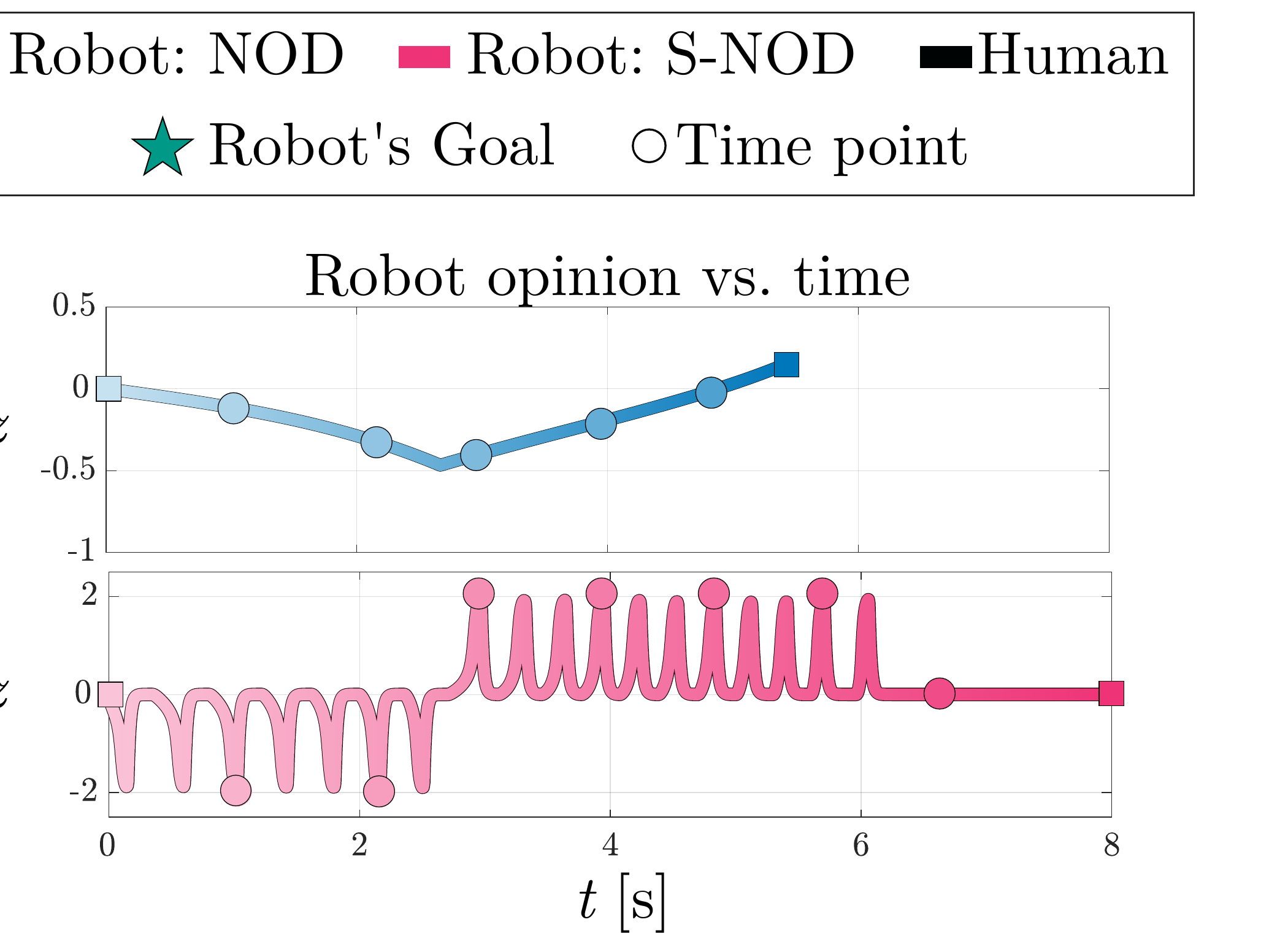}\label{subfig:intro_opinion_vs_time}}
    
\caption{
(a): Trajectory of a robot controlled with NOD (S-NOD) is shown with a blue (pink) line as it navigates towards a goal (star) in the presence of an oncoming human mover (black). The NOD robot experiences a collision; the S-NOD robot does not. (b): Opinion $z$ of the robot over time $t$. Circles denote matching time points along  trajectories and opinions. \\\textcolor{NavyBlue}{
Figures \ref{fig:intro}, \ref{fig:phase plane}, \ref{fig:multagent}, and \ref{fig:robot_trajectories} animated 
at \scriptsize{\url{https://spikingNOD.github.io}.}}}
\label{fig:intro}
\end{figure}

To derive S-NOD, we introduce in NOD 
a {\it slow} regulation term inspired by the dynamics of excitable (spiking) signal processing systems. 
The resulting excitable dynamics give S-NOD its superior agility in decision-making. Where NOD allows for a fast decision, S-NOD allows for autonomous fast sequential decision-making, not requiring any ad-hoc reset of the model state once a decision is made. Where NOD is flexible, 
 S-NOD is flexible and capable of fast ``changes of mind" and adaptive responses when the 
 information about the options changes rapidly and unexpectedly. 
 Further, S-NOD provides on-demand (event-based) opinion formation in the sense that large opinions are formed sparsely in time in the form of ``decision events'' and only when context requires it. 
 This makes S-NOD efficient. 
The agility of S-NOD is illustrated in Fig.~\ref{fig:intro} in the context of a robot navigating around a human mover as studied in~\cite{ref:cathcart2023proactive}. 
Our major contributions in this paper are the presentation, analysis, and illustration of 
a first-of-its-kind model of two-dimensional excitable (spiking) dynamics for decision-making over two options, which provides superior agility (fast, flexible and, adaptive behavior), especially important in changing contexts. Also, in Section~\ref{sec:Fast&Flexible}, we present a new analysis of the singularity in NOD for a single decision-making agent and two 
options. We prove how a single feedback gain $K_u$ tunes opinion formation. We show  
that when $K_u$ gets too large, an opinion can become so robust that it will not change quickly enough if a new input arrives in favor of the alternative option. 
In Section~\ref{sec: ENOD}, we present S-NOD for a single agent and two options. We show the existence of the spiking limit cycles associated with  excitable behavior. We show geometrically the agility in behavior. 
We generalize S-NOD to a network of multiple agents  and apply it to a social robot navigation problem in Section~\ref{sec: robot nav}. 

\section{Fast and Flexible Decision-Making: NOD}\label{sec:Fast&Flexible}

We recall NOD~\cite{ref:Bizyaeva_NODTunable2023,ref:Leonard_ARpaper2024} in Section~\ref{subsec:NOD definition}    for a single agent
evolving continuously over time its real-valued opinion about two mutually exclusive options 
with possible input present. 
 In Section~\ref{subsec:NOD analysis} we analyze stability of the neutral opinion solution and prove conditions on feedback gain $K_u$ that determine the type of singularity (type of pitchfork bifurcation) in the dynamics. We show that, by shaping bifurcation branches, $K_u$ tunes opinion formation. 
 In Section \ref{subsec:FF tunability}, we show limits on tunability of NOD that sacrifice agility in decision-making, motivating 
 S-NOD in Section~\ref{sec: ENOD}. 

\subsection{NOD for Single Decision-Maker and Two Options}\label{subsec:NOD definition}

We let an agent  represent a single decision-maker.
Let $z(t)\in \mathbb{R}$ define the agent's opinion at time $t$ about two mutually exclusive options. 
The more positive (negative) is $z$, the more the agent favors (disfavors) option 1 and disfavors (favors) option 2. When $z\!=\!0$, the agent is neutral about the two options, i.e., in a state of indecision. 
Let $u(t) \!\geq\! 0$ define the attention of the agent at time $t$ to its observations; $u$ is implemented as a gain in the dynamics.  Let $b(t) \in \mathbb{R}$ define an input signal at time $t$ that represents external stimulus and/or internal bias. When $b(t)\!>\!0$ ($b(t)\!<\!0$), it provides information (evidence) in favor of option 1 (option 2). 

 Decision-making variables $z$, $u$ evolve in continuous time $t$ according to the following NOD, adapted from \cite{ref:Bizyaeva_NODTunable2023,ref:Leonard_ARpaper2024}:
 \vspace{-.1truein}
 \begin{subequations}\label{eq:z_u_og}
    \begin{align}
        \tau_z \dot{z} &= -d \, z + \tanh \big( u \left( a z\right) + b \big) \label{eq:zdot_og}\\
        u &= u_0 + K_u z^2  \label{eq:u_og}
    \end{align}
\end{subequations}
where ${\dot z}:=\! dz/dt$. 
$\tau_z\!>\!0$ is a time constant, and damping coefficient 
$d\!>\!0$ weights the negative feedback on $z$ that regulates to the neutral solution $z\!=\!0$. The second term in \eqref{eq:zdot_og} provides a nonlinear positive feedback on $z$ with weight given by the product of $u$ and  amplification 
coefficient $a \!>\!0$, plus the effects of  $b$.  The saturation nonlinearity given by the $\tanh$ function enables fast-and-flexible decision-making through opinion-forming bifurcations~\cite{ref:Bizyaeva_NODTunable2023, ref:Leonard_ARpaper2024}. The positive feedback gain is state-dependent according to~\eqref{eq:u_og} and  grows with $z^2$. Hence, small deviations from the neutral solution ($z\!=\!0$) in response to small inputs leave attention low and do not trigger large, nonlinear opinion formation. Large enough deviations from the neutral solution in response to large enough inputs cause a sharp increase in attention and trigger large, nonlinear opinion formation. The resulting implicit threshold distinguishing small and large inputs is tuned by 
$u_0$ and $K_u$.

\subsection{Analysis of Single-Agent, Two-Option NOD}\label{subsec:NOD analysis}

We study the dynamics and stability of solutions of  system \eqref{eq:zdot_og}-\eqref{eq:u_og}
using bifurcation theory. A local bifurcation refers to a change in number and/or stability of equilibrium solutions of a nonlinear dynamical system as a (bifurcation) parameter  is changed. The state and parameter values at which this change occurs is the \textit{bifurcation point}. At a bifurcation point, one or more eigenvalues of the Jacobian of the model must have zero real part \cite{ref:golubitsky2012singularities1, ref:strogatz2000}, i.e., a bifurcation point is a singularity of the model vector field. 


Our main interest is in the 
\textit{pitchfork} bifurcations. There are two generic types of pitchforks. A \textit{supercritical pitchfork} bifurcation describes how one stable solution becomes unstable and two stable solutions emerge as the bifurcation parameter increases. A \textit{subcritical pitchfork} bifurcation describes how two unstable solutions disappear and one stable solution becomes unstable as the bifurcation parameter increases.

Our objective is to understand how thresholds of fast-and-flexible decision-making are controlled by the model parameters with the goal of designing feedback control laws for those parameters that can make decision thresholds adaptive to context. 
Substituting~\eqref{eq:u_og} into~\eqref{eq:zdot_og} yields
\begin{equation}\label{eq:NOD}
   \tau_z \dot{z} = -d \, z + \tanh \Big( (u_0 + K_u z^2)\cdot \left( a z\right) + b\Big).
\end{equation}
We first study~\eqref{eq:NOD} in the case  $b\!=\!0$, i.e., when there is no evidence to distinguish the options, and 
bifurcations are symmetric. 
Then, we  introduce $b\!\neq \!0$  and use unfolding theory \cite{ref:golubitsky2012singularities1} to understand the effects of inputs. 

\begin{lemma}\label{prop:u_0^*}
\textit{(NOD Taylor expansion and singularity)}:
Consider \eqref{eq:NOD} and let $b\!=\!0$. Then the 
solution $z\!=\!0$ is always an equilibrium, and the Taylor expansion of \eqref{eq:NOD} about $z\!=\!0$ is
\begin{multline}\label{eq:z_dot_taylorExpans1}
\dot{z} = 
\frac{1}{\tau_z}
    \Bigg(
    \left(a \, u_0-d\right)z
    + a \left(K_u-\frac{a^2 u_0^3}{3}\right) z^3\\
    +a^3 u_0^2 \left(\frac{2 a^2  u_0^3}{15}-K_u \right) z^5
    \Bigg) + \mathcal{O}(z^7).
\end{multline}
 A singularity exists at $(u_0,z) \!=\! (u_0^*,0)$, with $u_0^*\!=\!\frac{d}{a}$. 
 The 
 solution $z\!=\!0$  is stable (unstable) when $u_0\!<\!u_0^*$ ($u_0\!>\!u_0^*$).
 \end{lemma}
\begin{proof}
When $z\!=\!0$ and $b\!=\!0$, the right hand side of \eqref{eq:NOD} is zero thus $z\!=\!0$ is always an equilibrium.
We expand \eqref{eq:NOD} with  $b\!=\!0$ about $z\!=\!0$. The Taylor expansion of the hyperbolic tangent is $\tanh(w) \!=\! w - w^3/3 + 2 w^5/15 + \mathcal{O}(w^7)$. Substituting this into \eqref{eq:NOD} yields \eqref{eq:z_dot_taylorExpans1}. The Jacobian $J(z) \!=\! \frac{d \dot{z}}{d z}$ of \eqref{eq:z_dot_taylorExpans1} evaluated at $z\!=\!0$ is $J(0) \!=\! (a \, u_0 - d)/\tau_z$, which is singular when $u_0 \!=\! u_0^* \!=\! \frac{d}{a}$. When $u \!<\! u_0^*$, $J(0) \!<\! 0$ thus $z\!=\!0$ is exponentially stable. When $u \!>\! u_0^*$, $J(0) \!>\! 0$ thus $z\!=\!0$ is unstable.
\end{proof}
\raggedbottom
We next explore in Proposition~\ref{prop:Ku_pitchfork} and Fig.~\ref{fig:Kcomparison} the effect of parameter $K_u$ on the cubic and quintic terms of \eqref{eq:z_dot_taylorExpans1} and its role in determining the type of singularity at $(u,z)\!=\!(u_0^*,0)$.

\begin{proposition}\label{prop:Ku_pitchfork}
\textit{($K_u$ determines singularity type)}:
Let $b\!=\!0$, $u_0^* \!=\! \frac{d}{a}$. The singularity of dynamics~\eqref{eq:NOD} at $(u,z)\!=\!(u_0^*,0)$  as proved in Lemma~\ref{prop:u_0^*} corresponds to a \textit{supercritical} pitchfork bifurcation for $K_u \!<\! \frac{d^3}{3a}$, a \textit{quintic} pitchfork bifurcation for $K_u \!=\! \frac{d^3}{3a}$, and a \textit{subcritical} pitchfork bifurcation for $K_u \!>\! \frac{d^3}{3a}$.
\end{proposition} 

\begin{proof}
Denote $p\!=\! (K_u \!-\! \frac{d^3}{3a})$  and $q\!=\! (\frac{2d^3}{15a} \!-\! K_u)$ as the coefficients of $a z^3/\tau_z$  and $(a^3 u_0^2)z^5/\tau_z$  in \eqref{eq:z_dot_taylorExpans1} resp. at $u_0 \!=\! u_0^* \!=\! \frac{d}{a}$. When $K_u\!<\!\frac{d^3}{3a}$ ($K_u\!>\!\frac{d^3}{3a}$), then $p\!<\!0$ ($p\!>\!0$) and  (\ref{eq:z_dot_taylorExpans1}) is the normal form of the supercritical (subcritical) pitchfork bifurcation \cite{ref:golubitsky2012singularities1}. When $K_u\!=\!\frac{d^3}{3a}$, then $p\!=\!0$, $q\!<\!0$ and (\ref{eq:z_dot_taylorExpans1}) is the normal form of the quintic pitchfork, by 
recognition problem~\cite[Prop.~VI.2.14]{ref:golubitsky2012singularities1} and 
its $Z_2$-symmetric universal unfolding~\cite[Prop.~VI.3.4; Fig.~VI.3.3]{ref:golubitsky2012singularities1}.
\end{proof}

\begin{figure}
    \centering
    \subfloat[\label{subfig:K_bifurcations}]
    {\includegraphics[width=.5456\columnwidth]{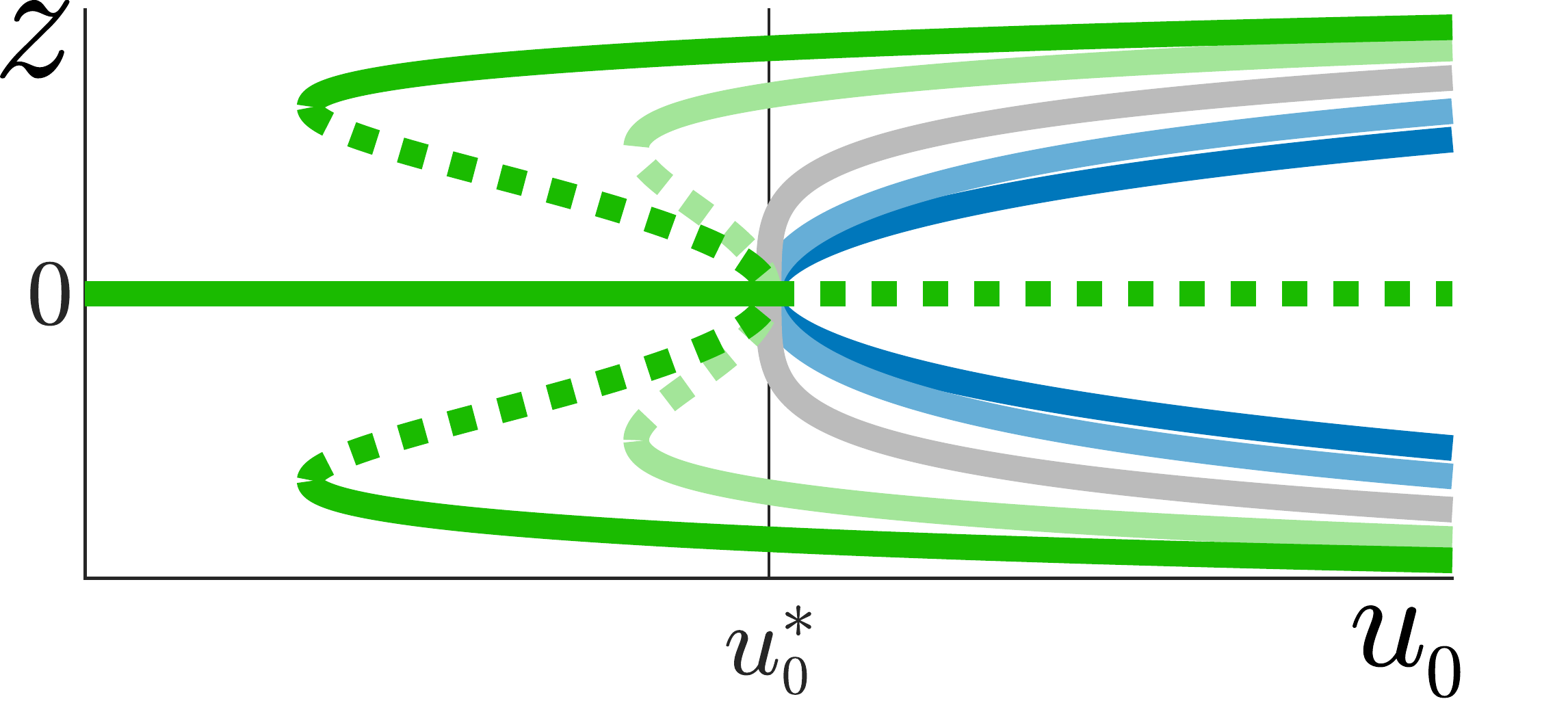}}
    \subfloat[\label{subfig:K_coefficientplots}]
    {\includegraphics[width=.4543\columnwidth]{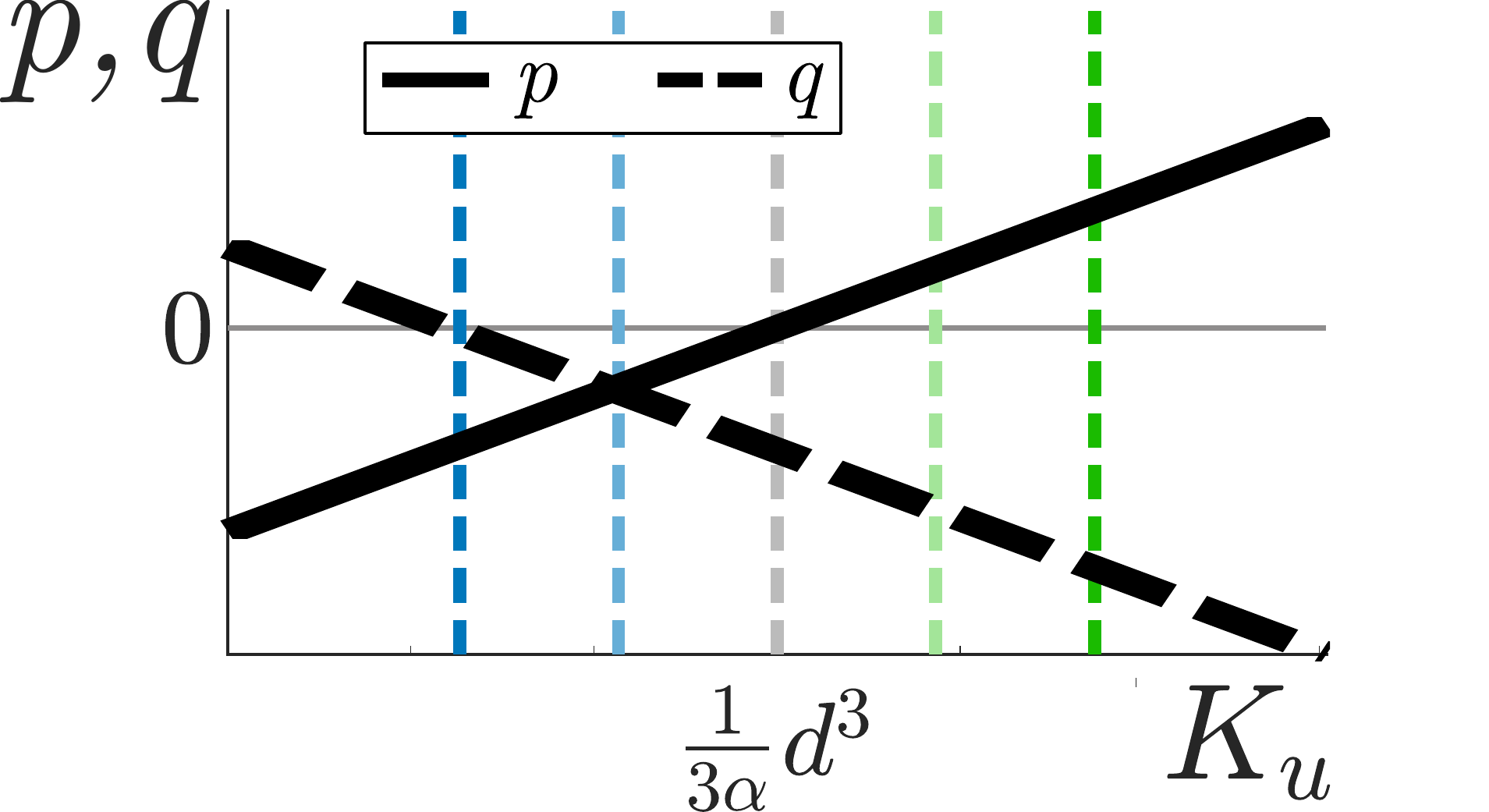}}
    \caption{The effect of $K_u$ on the bifurcation diagram of \eqref{eq:NOD} and the cubic and quintic terms of \eqref{eq:z_dot_taylorExpans1}. 
    (a): Bifurcation diagrams of 
    NOD \eqref{eq:NOD} 
    with $K_u$ values corresponding to the vertical dashed lines in (b). Stable (unstable) solutions are shown with solid (dotted) lines. The bifurcation point is $(u_0^*, 0)$. 
    (b): Coefficient $p$ ($q$) as a function of $K_u$ shown as a solid (dashed) black line. 
    }
    \label{fig:Kcomparison}
\end{figure}
Proposition~\ref{prop:Ku_pitchfork} uncovers the key role of $K_u$ in tuning opinion formation (Fig.~\ref{fig:Kcomparison}): (i) $K_u$ controls the supercritical vs. subcritical nature of opinion formation, and (ii) increasing $K_u$ increases opinion strength of non-neutral solutions. 
Bifurcation diagrams in Fig.~\ref{subfig:K_bifurcations} 
plot equilibrium solutions of~\eqref{eq:NOD}, i.e., solutions of $\dot z \!=\! 0$, as a function of bifurcation parameter $u_0$ for 
different values of $K_u$ (Fig.~\ref{subfig:K_coefficientplots}). 
The singularity 
at $(u_0,z)\!=\!(u_0^*,0)$ is a pitchfork bifurcation: blue, gray,  green lines show supercritical, quintic, subcritical solutions, respectively. 
For all $K_u$: when $u_0\!<\!u_0^*$, the neutral solution is stable; when $u_0\!\!>\!\!u_0^*$, the neutral solution is unstable and there is a bistable symmetric pair of solutions. 
When $K_u\!\!<\!\!\frac{d^3}{3a}$, 
the pitchfork is supercritical: there are no other solutions. 
When $K_u\!\!>\!\!\frac{d^3}{3a}$, 
the pitchfork is subcritical:
two stable non-neutral solutions appear for $u_0 \!<\! u_0^*$, through 
saddle-node bifurcations. As $K_u$ grows more positive, these solutions emerge for smaller values of $u_0$ and increase in magnitude, reflecting increasing opinion strength.
\begin{figure}[t!]
\vspace*{7pt}
    \centering
    {\includegraphics[width=0.98\columnwidth]{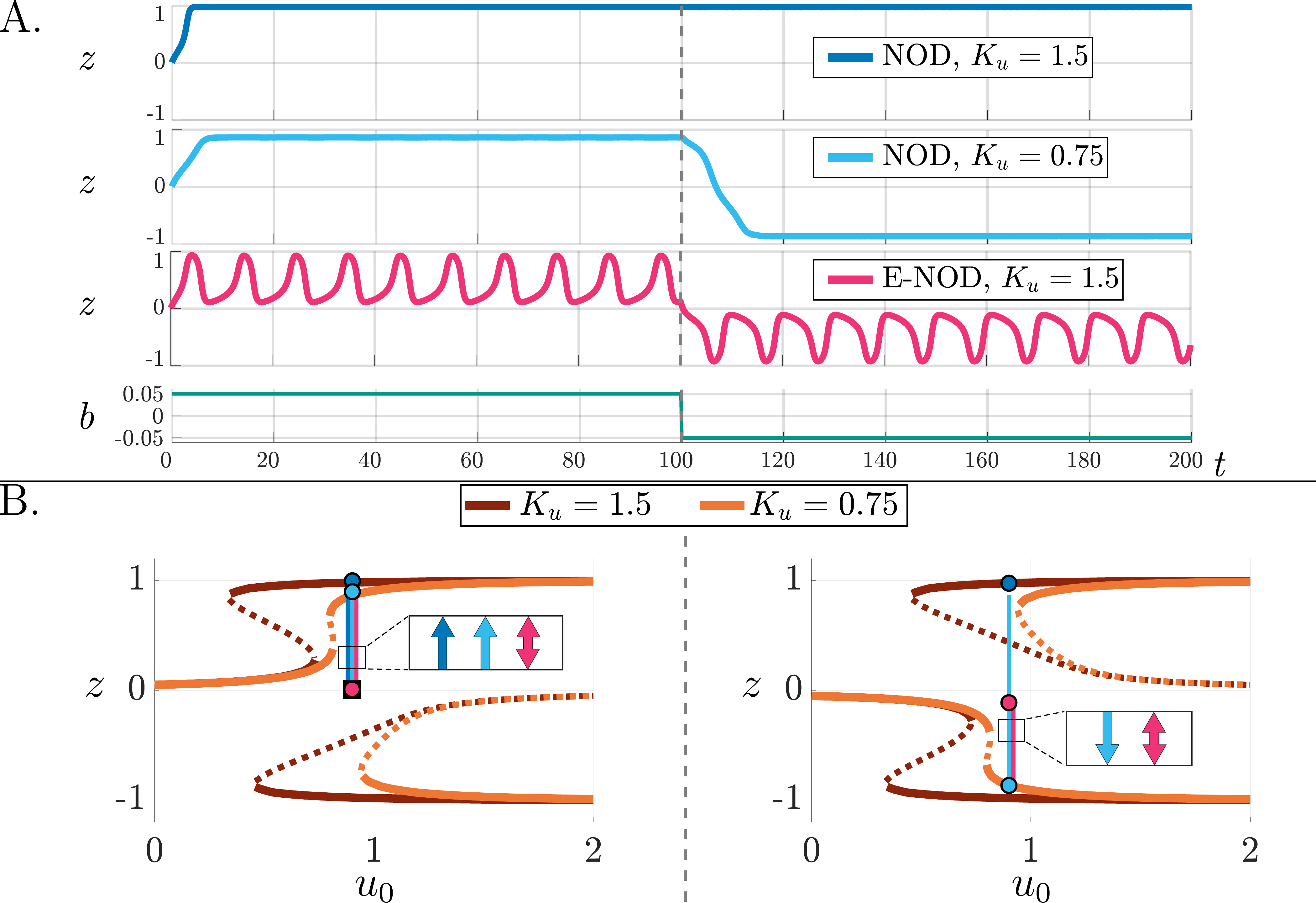}}
    \caption{Opinion solutions of NOD and S-NOD over time and associated bifurcation diagrams. 
    (A): Trajectories of NOD \eqref{eq:NOD} and S-NOD \eqref{eq:ENOD} for initial condition $(z, u_0)|_{t=0} = (0.01,0.9)$ for two $K_u$ values. Input signal $b$ is also shown over time.
    (B): Bifurcation diagrams of \eqref{eq:NOD} for the two values of $K_u$, showing the solutions in A. moving from initial state to steady state.}\label{fig:exNOD_vs_NOD} 
    \end{figure} 

\subsection{Limitation on Tuning of NOD}\label{subsec:FF tunability}
We prove that the region of multi-stability in the subcritical bifurcation of NOD grows as $K_u$ gets large. 

\begin{proposition}
\label{prop:monotonic}
    \textit{($K_u$ determines region of multi-stability)}: 
    Let $b\!=\!0$ and $(u_0^\dag, z^\dag)$ be either of the two saddle-node bifurcation points of the subcritical pitchfork of NOD \eqref{eq:NOD} for $K_u \!>\! \frac{d^3}{3a}$. Then  $u_0^\dag$ is a monotonically decreasing function of $K_u$, i.e., $\frac{\partial u_0^\dag}{\partial K_u}\!<\!0$.
\end{proposition}
\begin{proof}
    Let $K_u^\dag\!>\!\frac{d^3}{3a}$  and $f(z,K_u,u_0)\!:=\!-d\,z\!+\!\tanh(a z(u_0\!+\!K_u z^2))$. By hypothesis,   $f(z^\dag,K_u^\dag,u_0^\dag)\!=\!0$. We have $\frac{\partial f}{\partial u_0}(z^\dag,K_u^\dag,u_0^\dag)=a z^\dag\tanh'(a z^\dag(u_0^\dag\!+\!K_u^\dag (z^\dag)^2))\!\neq\! 0$, since $z^\dag \!\neq\! 0$. Following  \cite{ref:amorim2024threshold}, we use the implicit function theorem to show the existence of $g\!:\!\mathbb{R}^2\!\to\!\mathbb{R}$ such that for some neighborhood of $(z^\dag,K_u^\dag)$, 
    $f(z,K_u,g(z,K_u))\!=\!0$. 
    We get
    \begin{align*}
        \frac{\partial g}{\partial K_u} &= -\left ( \frac{\partial f}{\partial u_0} \right )^{-1} \left ( \frac{\partial f}{\partial K_u} \right ) \\
        &= -\frac{a (z^\dag)^3 \tanh'(a z^\dag(u_0^\dag+K_u^\dag (z^\dag)^2))}{a z^\dag \tanh'(a z^\dag(u_0^\dag+K_u^\dag (z^\dag)^2))} = -(z^\dag)^2.
    \end{align*}
    Since $z^\dag\!\neq\! 0$, 
    $u_0^\dag$ is monotonically decreasing in $K_u$. 
\end{proof}
\begin{figure*}[ht!]
    \centering
    \includegraphics[align=c,width=0.06696\linewidth]{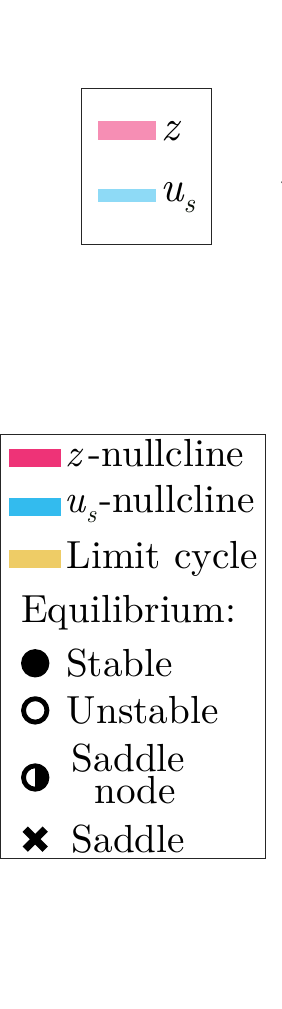}
    \hfill
    \subfloat[$u_0 = 0.1, b = 0$\label{subfig:u0_low}]
    {\includegraphics[align=c,width=.23\linewidth]{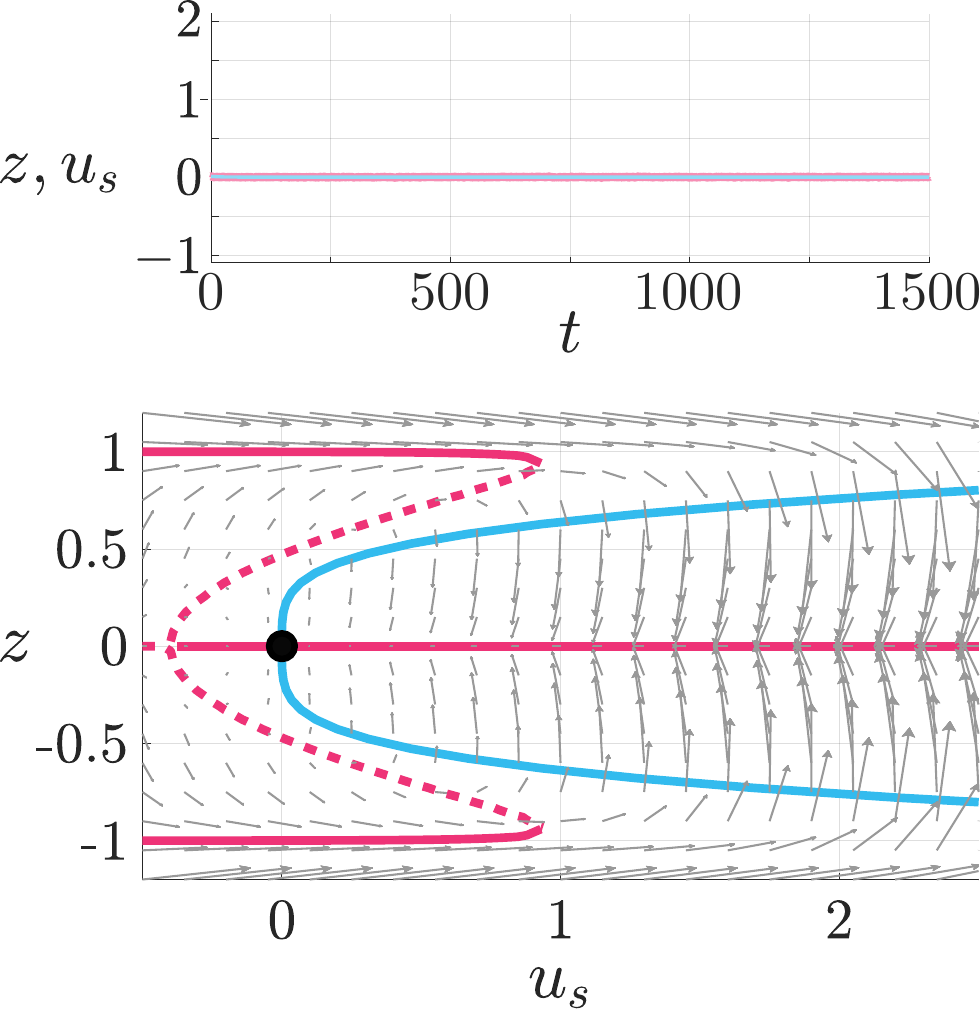}}
    \hfill
    \subfloat[$u_0 = 0.5, b = 0$\label{subfig:u0_med}]
    {\includegraphics[align=c,width=.23\linewidth]{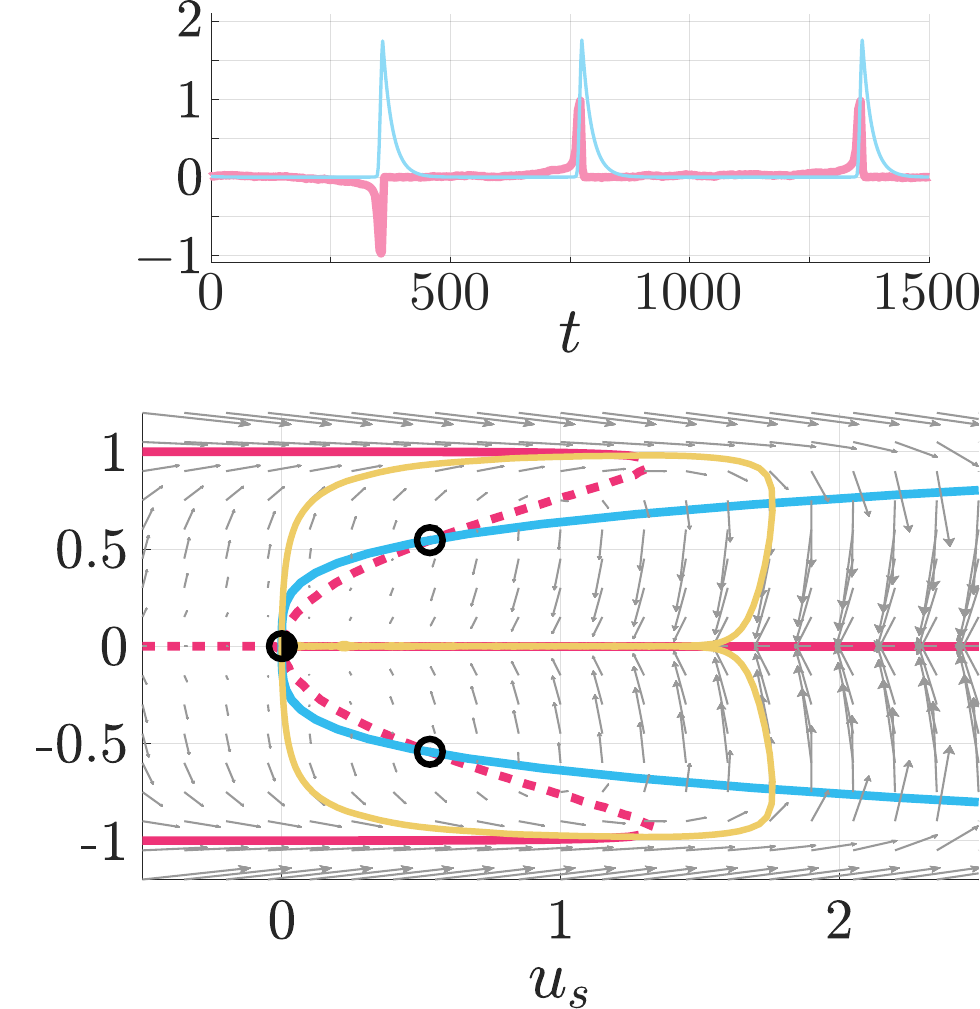}}
    \hfill
    \subfloat[$u_0 = 0.7, b = 0$\label{subfig:u0_high}]
    {\includegraphics[align=c,width=.23\linewidth]{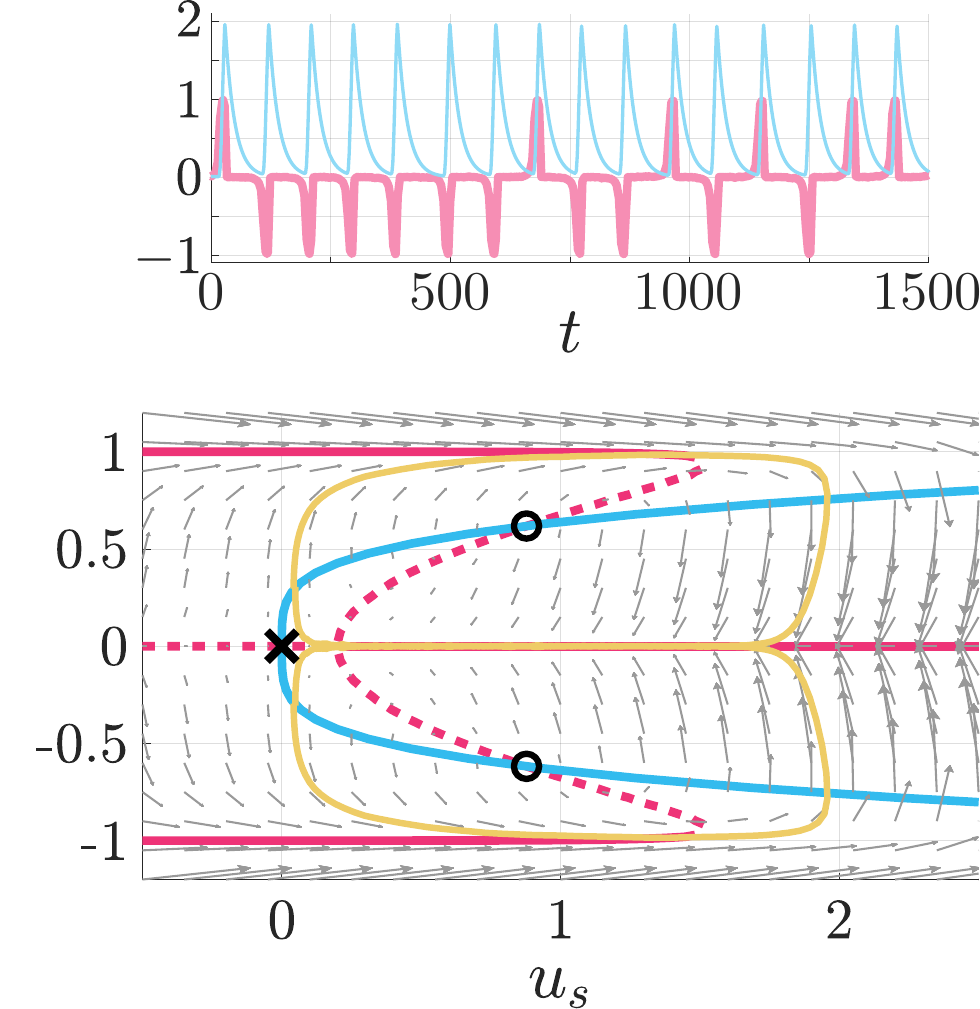}}
    \hfill
    \subfloat[$u_0 = 0.7, b = 0.1$\label{subfig:u0_high_b}]
    {\includegraphics[align=c,width=.23\linewidth]{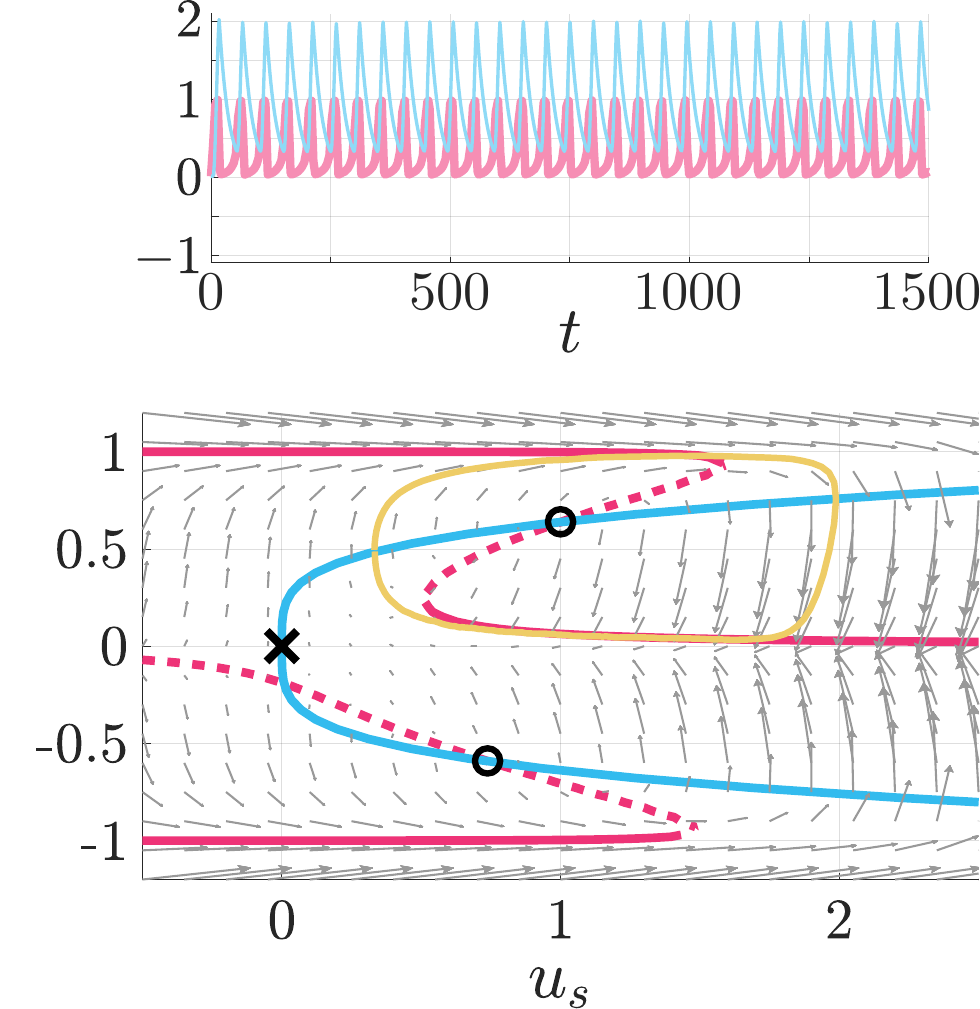}}
    \caption{The system solutions and $(u_s,z)$ 
    phase portrait
    as the basal attention $u_0$ increases. For all, $d \!= \!1, \alpha \!=\! 2,$ (thus $u_0^*\!=\! 0.5$)$, K_u \!=\! 2, K_{u_s} \!=\! 6, \tau_{u_s}/\tau_z \!=\! 10$.
    (Top): Example  solutions of $u_s$ and $z$ over time, with initial condition as $(u_s, z)|_{t=0} = (0.01, 0.01)$ and additive Gaussian distributed white noise. 
    (Bottom): The $u_s$-nullcline ($z$-nullcline) is shown as a blue (pink) line. Solid (dotted) lines indicate stable (unstable) branches of the $z$-nullcline with respect to
    ~\eqref{eq:z_dot_ENOD}. Gray arrows denote the vector field. Black filled circles show stable equilibria, unfilled are unstable equilibria, partially filled is a saddle-node bifurcations. Crosses show saddle equilibria. Saddle-node-homoclinic cycles in (b), and limit cycles in (c) and (d) are in yellow. \squeezeup} 
    \label{fig:phase plane}
\end{figure*}
Proposition~\ref{prop:monotonic} implies that one limitation of 
NOD is that large $K_u$ can make the region of multi-stability 
of NOD \eqref{eq:NOD} so large and robust that solutions can get ``stuck'' in one of the decision attractors unless very large inputs in favor of another decision state are applied. This is illustrated in Fig.~\ref{fig:exNOD_vs_NOD}A, where the first (dark blue) and second (light blue) NOD differ only in their $K_u$ parameters ($K_{u_1}\!>\!K_{u_2}\!>\!\frac{d^3}{3a}$) but their solutions are distinctively different. 
At the stimulus onset ($b\!>\!0$ for $0\!\leq\!t\!<\!100$),  the solution of the first NOD converges to $z\!>\!0$ much more rapidly than 
that of the second NOD. However, when the input switches 
values 
($b\!<\!0$ for  $t\!\geq\!100$), the solution of the first NOD gets stuck at a positive value, whereas the solution of the second NOD is able to track the change in input sign. This example reveals a fundamental trade-off between speed/robustness 
(first, dark blue NOD) and flexibility (second, light blue NOD).

Instead of aiming to fine-tune the gain $K_u$ around a hard-to-define fast/robust enough yet flexible enough decision-making behavior, we use mixed-feedback principles to make the system {\it excitable}, inheriting both the speed of large-$K_u$ NODs and the flexibility of small-$K_u$ NODs and imparting system agility. The behavior of the resulting S-NOD is shown in 
pink in 
Fig.~\ref{fig:exNOD_vs_NOD}A. By generating ``decision spikes'' the S-NOD is as fast as the high $K_u$ NOD and as flexible as the low $K_u$ NOD. In what follows, we present the S-NOD model, its analysis, and its multi-agent generalization.



 


\section{Agile Decision-Making: S-NOD}
\label{sec: ENOD}
We present and analyze the Spiking Nonlinear Opinion Dynamics (S-NOD) model.  

\subsection{S-NOD for a Single Agent and Two Options}\label{subsec:ENOD definition}
We define S-NOD by introducing 
a slow regulation variable $u_s$ to NOD \eqref{eq:NOD}, as in the fast-positive, slow-negative mixed-feedback structure of excitable systems~\cite{dayan2005theoretical,Balazsi2011,suel2006excitable,fromm2007electrical}:
\begin{subequations}\label{eq:ENOD}
    \begin{gather}
        \tau_z \dot{z} = -d \, z + \tanh \Big( (u_0- u_s + K_u z^2)\cdot \left( a z\right)  + b \Big) \label{eq:z_dot_ENOD} \\
        \tau_{u_s} \dot{u}_s = K_{u_s} z^4 - u_s \label{eq:u_s_dot_ENOD}
    \end{gather}
\end{subequations}
where $\tau_{u_s} \!\gg\! \tau_z$ is larger by at least an order of magnitude 
such that $u_s$ responds more slowly than $z$. 
S-NOD as in~\eqref{eq:ENOD} describes dynamics with excitability: a fast positive feedback (mediated by $z$) acts to excite the system, while a slow negative feedback (mediated by $u_s$) regulates it back to near the ultrasensitive pitchfork singularity (as seen in Fig.~\ref{fig:Kcomparison}). 

The fast positive feedback in~\eqref{eq:z_dot_ENOD} comes from the term $(u_0+ K_u z^2)(a z)$, as in NOD~\eqref{eq:NOD}. 
The slow negative feedback comes from the coupling of the new variable $u_s$. As $z$ grows in magnitude (whether positive or negative, i.e., independent of which option is preferred), $u_s$ also grows, although more slowly, according to~\eqref{eq:u_s_dot_ENOD}. When $u_s$ becomes large, it drives $z$ back towards zero because of the $-u_s(a z)$ term in ~\eqref{eq:z_dot_ENOD}. 
As $z$ decreases towards zero, $u_s$ also decreases back towards zero according to~\eqref{eq:u_s_dot_ENOD}.  
The result is a spike in 
$z$ and $u_s$. This is a limit cycle, which bounds $z$ and $u_s$, and is analyzed next.

The S-NOD equations~\eqref{eq:ENOD} can be generalized to $N_o \geq 2$ options analogously to the generalization of NOD to $N_o$ options~\cite{ref:Bizyaeva_NODTunable2023,bizyaeva2023multi}. We will investigate this in future work.

\subsection{Geometric Analysis of Single-Agent, Two-Option S-NOD}\label{subsec:ENOD analysis}

We use phase-plane analysis to study and illustrate the spiking and decision-making behavior of S-NOD~\eqref{eq:ENOD}. 
To construct the phase-plane, we first compute the nullclines. 




The $z$-nullcline is defined as the solution pairs $(u_s,z)$ that satisfy $\dot z \!=\! 0$ for~\eqref{eq:z_dot_ENOD}. This is 
equivalent to solving for the equilibrium solutions of~\eqref{eq:NOD} as a function of $u_0$ as in Section~\ref{subsec:NOD analysis}.
Thus, the $z$-nullcline (pink in Fig.~\ref{fig:phase plane}) 
is analogous to the bifurcation diagram of~\eqref{eq:NOD}, mirrored about the vertical axis and shifted right by $u_0$, with $u_s^*\!=\!u_0\!-\!d/a \!=\!u_0\!-\!u_0^*$. When $b\!=\!0$, the neutral solution $z\!=\!0$ is a stable (unstable) equilibrium of~\eqref{eq:z_dot_ENOD} for $u_s \!>\! u_s^*$ ($u_s \!<\! u_s^*$). 
The $u_s$-nullcline (blue in Fig.~\ref{fig:phase plane}) is defined as the solution pairs $(u_s,z)$ that satisfy $\dot{u}_s \!=\! 0$ in \eqref{eq:u_s_dot_ENOD}, which gives the quartic parabola 
$u_s \!=\! K_{u_s} z^4$. 
The larger $K_u$, the more narrow this parabola is.

The intersections of the nullclines 
determine equilibrium solutions of S-NOD~\eqref{eq:ENOD} 
and, as shown in Fig.~\ref{fig:phase plane}, 
depend on the value of $u_0$. 
 If $b\!=\!0$, the \textit{neutral solution} $(u_s, z) \!=\! (0,0)$ is always an equilibrium. Two more equilibria, symmetric about $z\!=\!0$, may be present for high enough $u_0$. 


Fig.~\ref{subfig:u0_low} depicts the phase-plane when $b\!=\!0$ and $u_0\!<\!u_0^*$. The nullclines have 
one point of intersection at the neutral solution. The neutral solution is stable. Trajectories will converge to and settle at this point and no excitable behavior in the decision-making will take place. 


Fig.~\ref{subfig:u0_med} depicts the phase-plane when $b\!=\!0$ and $u_0 \!=\! u_0^*$. The nullclines have three points of intersection, the neutral solution and two unstable equilibria symmetric about $z\!=\!0$. The neutral solution is a saddle-node bifurcation with one exponentially stable eigendirection 
(along $z\!=\!0$) and one marginally unstable eigendirection 
(along $u_s\!=\!0$). There are two saddle-node-homoclinic (infinite period) cycles, 
diverging either upward and downward from the saddle-node. In the absence of noise/exogenous perturbations, all trajectories asymptotically converge to the neutral solution. The presence of any arbitrarily small noise makes the trajectory escape from the neutral solution at random time instants along either the upward or downward saddle-node-homoclinic cycle, leading to large prototypical excursions in the $(u_s,z)$ plane. 

These large prototypical excursions 
resemble of ``spiking'' trajectories of excitable neuronal system. By analogy, we call them ``decision spikes'' or ``excitable decisions''. In contrast to neuronal spikes which happen in only one direction, decision spikes can happen in 
as many directions as there are options. For the one-dimensional two-option dynamics studied here, both upward (in favor of option 1) and downward (in favor of option 2) decision spikes are possible.

Fig.~\ref{subfig:u0_high} depicts the phase-plane when $b\!=\!0$ and $u_0 \!>\! u_0^*$. 
Three fixed points are unstable and it is possible to prove, along the same lines as~\cite{ref:Ozcimder_ACCdance2016_2investigating}, the existence of two limit cycles, 
symmetric about the horizontal axis $z=0$. These limit cycles are made of repetitive decision spikes, i.e., spiking decision limit cycles.  Geometric singular perturbation analysis~\cite{krupa2001relaxation,ref:arbelaiz2024excitablecrawling} provides the tools to rigorously prove the existence of these spiking decision limit cycles. Such an analysis goes beyond the scope of this paper. Instead, we leverage  Fig.~\ref{subfig:u0_high} to describe qualitatively a typical oscillatory spiking decision behavior in the presence of small noisy perturbations.

Consider Fig.~\ref{subfig:u0_high} and a trajectory with a large initial condition in $u_s$. Initially, the trajectory rapidly converges to the $z=0$ axis, then slowly slides leftward approaching the 
neutral solution. 
As soon as the trajectory hits this 
equilibrium, noisy perturbations push it either upward or downward, generating an upward or downward decision spike respectively. The decision spike trajectory brings the trajectory back to the pitchfork singularity, the next decision spike is generated, and the spiking decision cycle continues.


When input 
$b \!\neq\! 0$ and $|b|$ is sufficiently large, the $z$-nullcline {\it unfolds} accordingly to the universal unfolding of the pitchfork~\cite[Ch.~III]{ref:golubitsky2012singularities1}. Due to the nullcline unfolding, the phase-plane geometry changes qualitatively as shown in Fig.~\ref{subfig:u0_high_b}. Similar geometric singular perturbation analysis methods as those employed for the analysis of Fig.~\ref{subfig:u0_med} and \ref{subfig:u0_high} reveal the existence of a {\it unique} spiking decision limit cycle associated to spiking decisions toward the option favored by the inputs (
e.g., upward decision spikes in the case of Fig.~\ref{subfig:u0_high_b} where $b>0$ provides evidence in favor of option 1).

Observe that in the presence of informative inputs (Fig.~\ref{subfig:u0_high_b}), the decision spiking frequency is higher than in the case of endogenous decision spiking oscillations (Fig.~\ref{subfig:u0_high}). This feature is similar to spike frequency indicating input intensity in 
neural systems. In applications like robot navigation, $u_0$ can be controlled to avoid endogenous spiking.


\section{Agile Multi-Agent Decision-Making: S-NOD}\label{sec: robot nav}

\subsection{{S-NOD for Multiple Agents and Two Options}}
\begin{figure}
    \centering
    \subfloat{{\includegraphics[ width=.5\columnwidth]{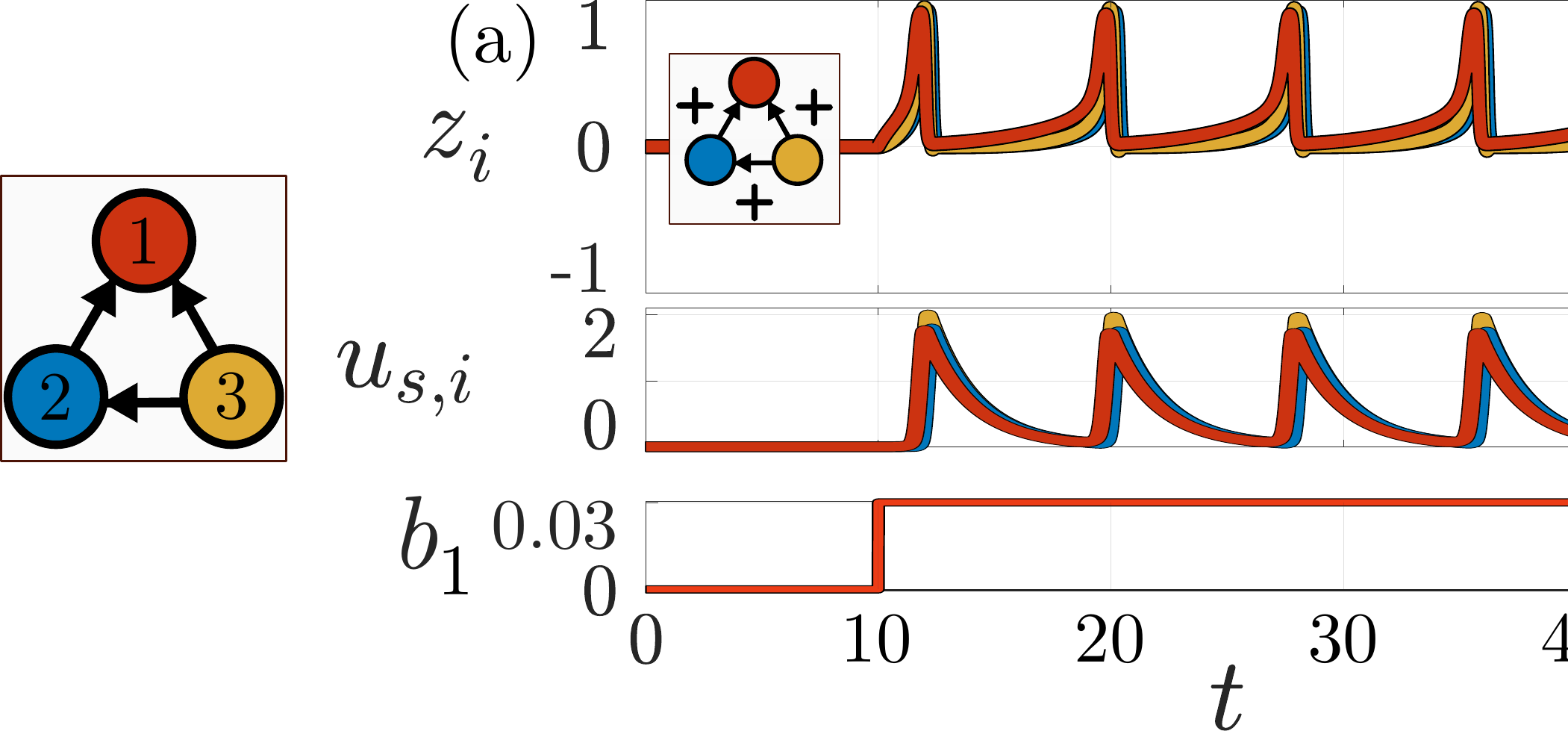}}\label{subfig:multagent_agreement}}
    \subfloat{{\includegraphics[ width=.5\columnwidth]{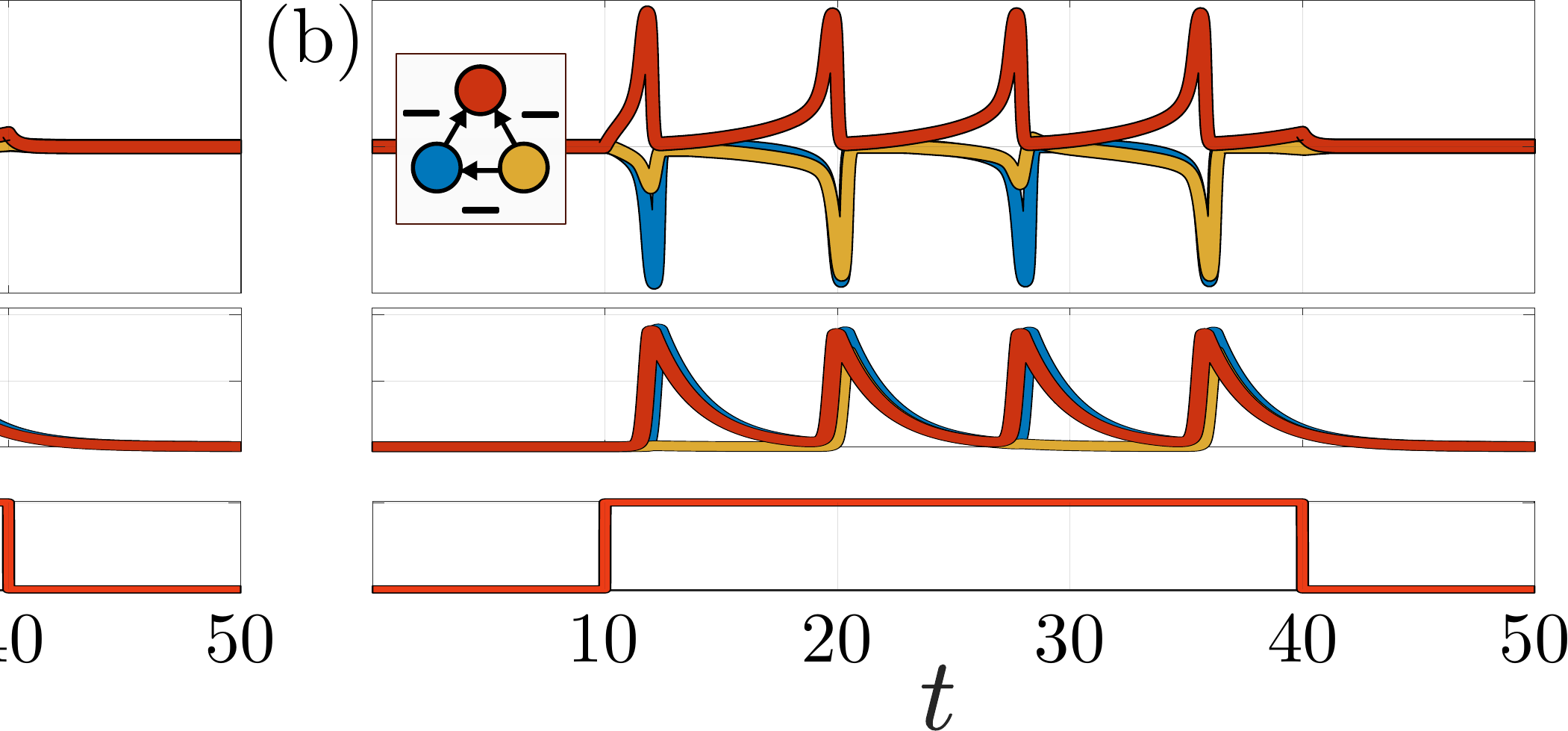}}\label{subfig:multagent_disagreement}}
    \caption{Response of three agents with the same communication network but for opposite signs on graph edge weights and an input $b_1$ applied only to agent $1$. 
  For $i\neq k$, (a) $a_{ik}= + 0.1$ and (b) $a_{ik}= - 0.1$. Parameters are $a_{ii}=1$, $d=1$, $K_u=2.3$, $K_{u_s}=16$, $u_0=0.9$, $\tau_{u_s}/\tau_{z}=20$.}
    \label{fig:multagent}
\end{figure}
We can generalize the single-agent S-NOD equations~\eqref{eq:ENOD} to the case of $N_a$ agents in the same way that NOD generalizes to $N_a$ agents \cite{ref:Bizyaeva_NODTunable2023}. The multi-agent NOD models the decision-making process of multiple agents sharing and influencing one another's opinions over a communication network. Examples include agents choosing how they distribute their time over two resource patches or deciding whether to move to the left or right when navigating a cluttered space, all while integrating information from other agents' opinions. In the multi-agent S-NOD, each agent $i$ has two state variables $z_i$ and $u_{s,i}$ with dynamics given by 
\begin{subequations}\label{eq:MA_ENOD}
\begin{gather} 
        \begin{adjustbox}{max width=0.88\columnwidth}
        $
          \textstyle \tau_z \dot{z}_i \!=\! -d \,z_i \!+\! \tanh \!\left(\!(\!u_0\!-\!u_{s,i}\!+\!K_{u}z_i^2)\!\left(\!\sum_{k=1}^{N_a} \!a_{ik}z_k \!\right) +\!b_i\! \right) \label{eq:MA_NOD}
        $
        \end{adjustbox}
\\
    \tau_{u_s} \dot{u}_{s,i} = -u_{s,i} + K_{u_{s}}z_i^4 \label{eq:MA_u_s_dot}
\end{gather}    
\end{subequations}
where $A \!=\! [a_{ik}]\in\mathbb{R}^{N_a \times N_a}$ is the S-NOD network adjacency matrix, capturing both the strength ($a_{ii}\geq 0$) of a self-reinforcing term
 and the strength ($|a_{ik}|$) of the influence of the opinion of agent $k$ on the opinion of agent $i$. 
 If $a_{ik}$ is positive (negative), then an opinion of agent $k$ in favor of one of the options influences an opinion of agent $i$ in favor of the same (other) option.  
We assume homogeneous agents, i.e., all agents have the same $d$, $ u_0$, $K_u$ and $K_{u_s}$. 
S-NOD as presented in~\eqref{eq:MA_ENOD} is the networked, distributed version of 
\eqref{eq:ENOD}. 

In Fig.~\ref{fig:multagent} we simulate the opinion dynamics of~\eqref{eq:MA_ENOD} for three agents in a network when agent 1 receives a constant input $b_1>0$ for $t \in [10,40]$.  
For the loop-free networks of Fig.~\ref{fig:multagent}, we see that when the weights $a_{ik}$ are positive (negative) the spiking of agent 1 for option 1 triggers synchronized (anti-synchronized) spiking of agents 2 and 3.  
Future work will characterize different possible behaviors, e.g., opinion spike (anti)synchronization, for classes of network structures.

\subsection{Application to Social Robot Navigation}
\begin{figure}
    \centering
    \subfloat{\includegraphics[width=.333\columnwidth]{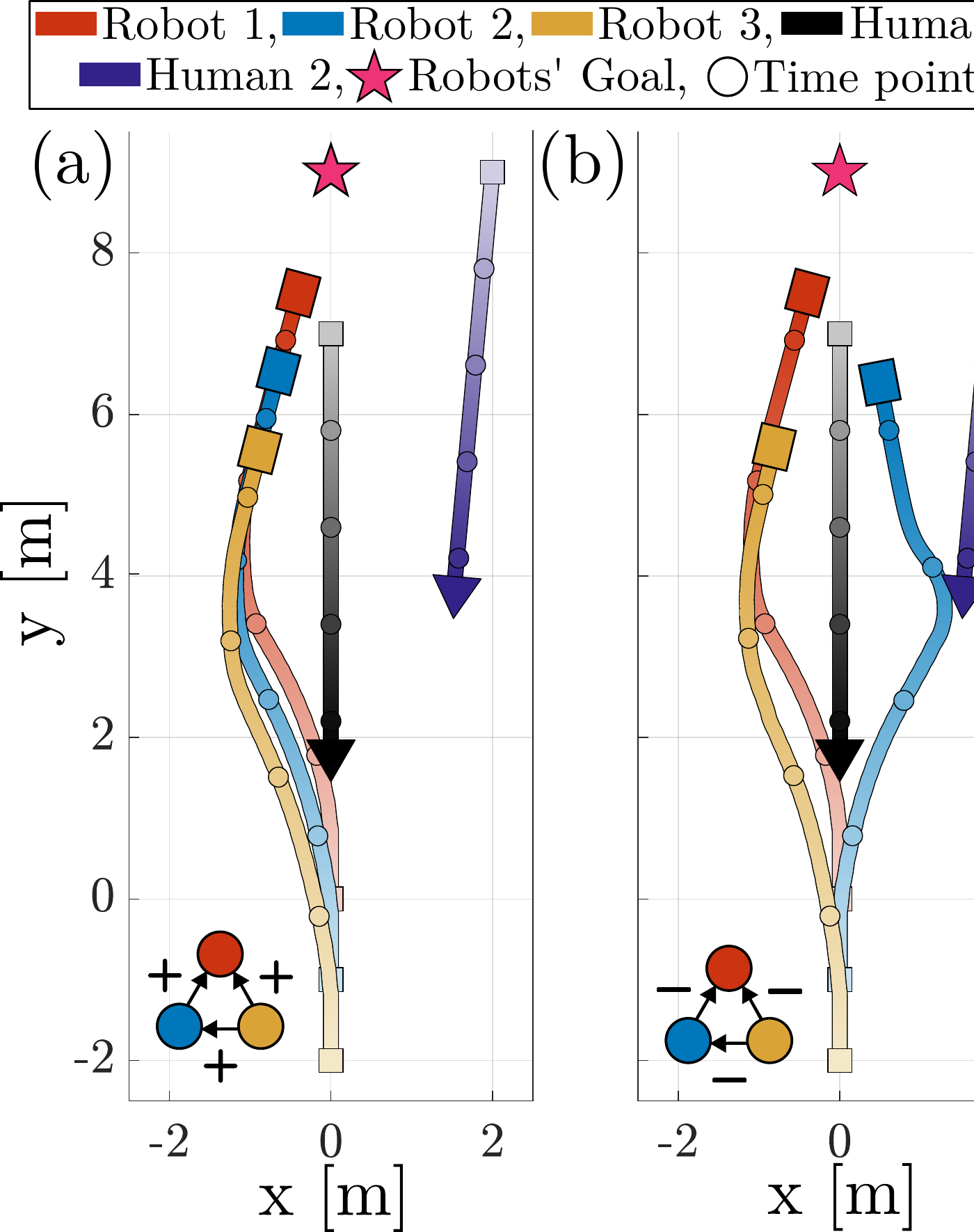}\label{subfig:robot_traj_agreement}}
    \subfloat{\includegraphics[width=.333\columnwidth]{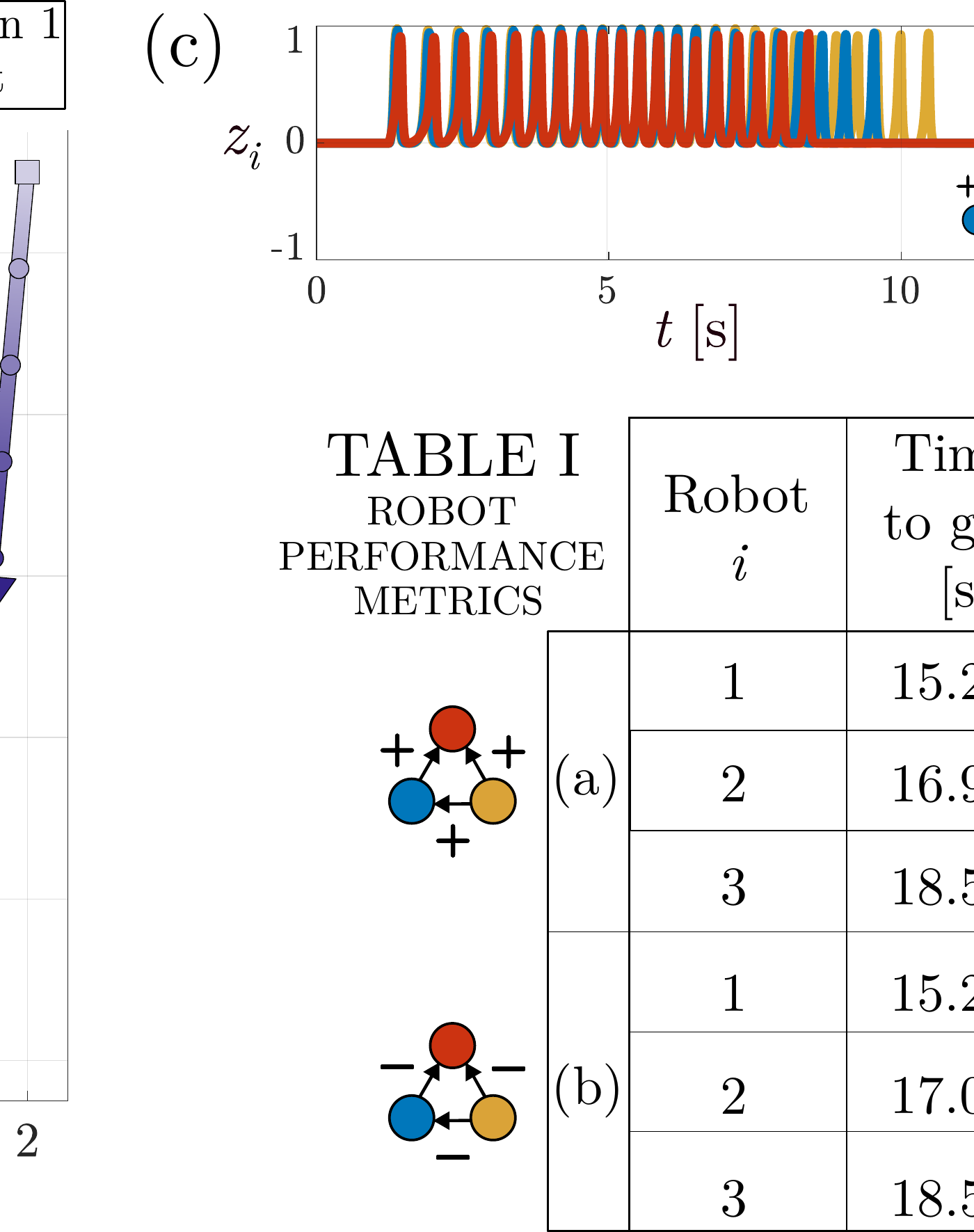}\label{subfig:robot_traj_disagreement}}
    \subfloat{\includegraphics[width=.333\columnwidth]{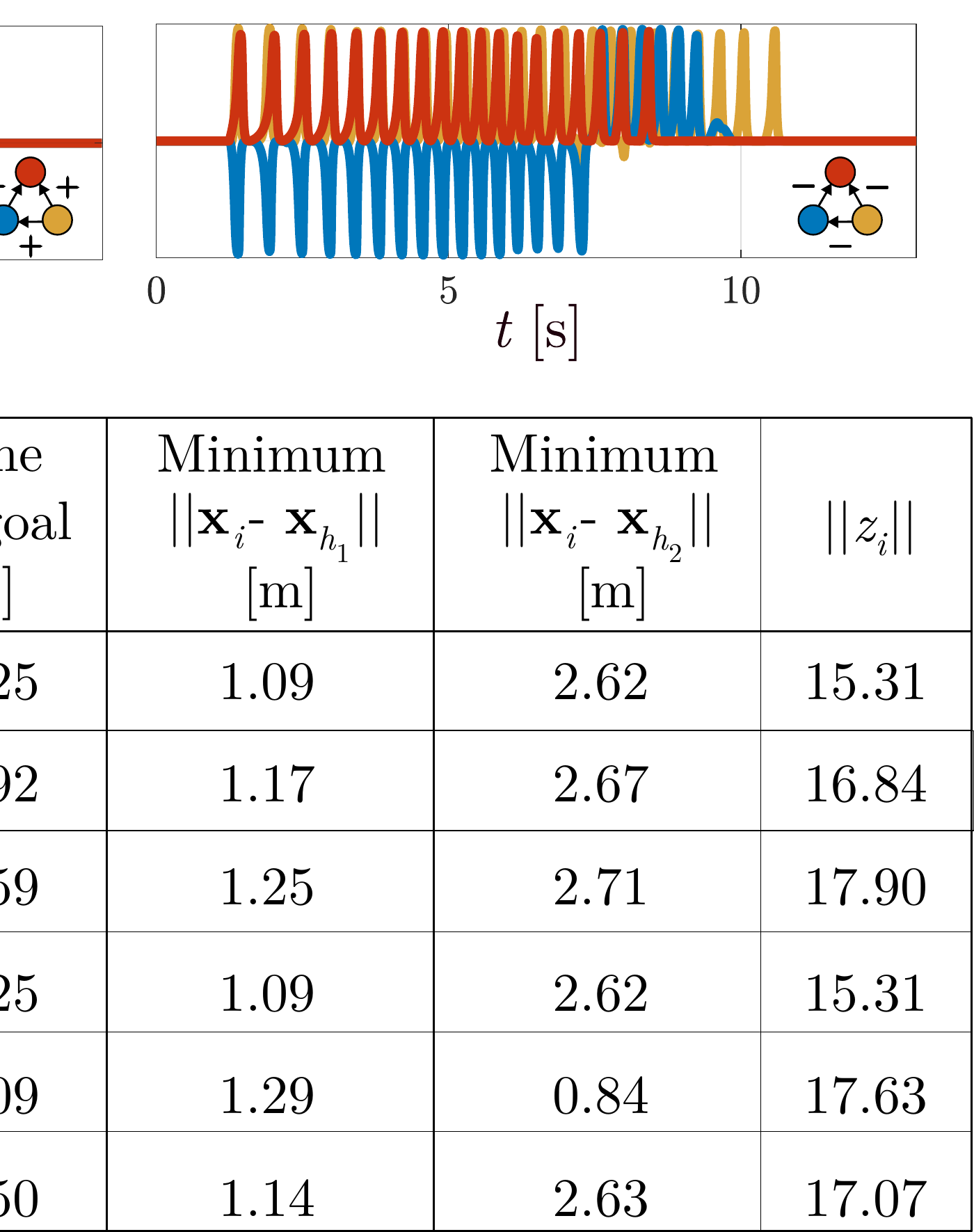}\label{subfig:robot_traj_opinions}}
    \caption{Trajectories of social robot 
    teams navigating around approaching human movers. Communication network and parameters are as in Fig.~\ref{fig:multagent} with (a) $a_{ik}= + 0.1$  and (b) $a_{ik}= - 0.1$. (c) Plots of $z$ over time $t$ for the robots. Table I lists performance metrics of the robots.}
    \label{fig:robot_trajectories}
\end{figure}
We use S-NOD~\eqref{eq:MA_ENOD} to design a decentralized, agile controller for social robots navigating around oncoming human movers in 2D. 
Each robot has a nominal control that steers it toward its goal by regulating its heading direction through proportional negative feedback. 
The S-NOD state $z_i$ defines the strength of robot $i$'s preference for turning left ($z_i\!>\!0$) or right ($z_i\!<\!0$). 
A term mediated by $z_i$ is added as positive feedback to the control on steering behavior. This overcomes negative feedback regulation and promotes fast reactive steering when possible collisions with oncoming human movers are imminent. 
Simulations of the resulting navigation behavior are in Fig.~\ref{fig:intro} and~\ref{fig:robot_trajectories} and  animated at \textcolor{NavyBlue}{\small{\url{https://spikingNOD.github.io}}}. 


  
   To anticipate collisions, robot $i$ can estimate its distance ($\rho_i$) to a human, bearing angle ($\eta_i$) on the human, and angle ($\eta_{h,i}$) between the human's heading direction and the robot-human vector. Robots can exchange steering opinions over a communication network as in \eqref{eq:MA_NOD}.
   The robot's attention grows above its basal level $u_0$ as collision risk grows with decreasing $\rho_i$ and $\eta_i$. This increases the strength of (i) the positive feedback loop of the steering controller, and (ii) the interactions with other robots to achieve coordinated obstacle avoidance. 
   Thus, each robot's steering opinion deviates from navigating toward its goal dependent on $\rho_i$, $\eta_i$, $\eta_{h_i}$, and on other robot opinions. 
   We let $\hat{z}_h(t) \!=\! \tan(\eta_h(t))$ be a proxy for the human's 
   opinion at time $t$ and add $\hat{z}_h$ to the term $\sum_{k=1}^{N_a} \!a_{ik}z_k$ in \eqref{eq:MA_NOD}. 
    Coordination among robots derives from the sign of $a_{ik}$: when $a_{ik}\!>\!0$ ($a_{ik}\!<\!0$), robot $i$ is influenced to make a similar (opposite) steering choice as robot $k$. We let $|a_{ik}|$ decay with growing distance between robots $i$, $k$. 

Fig.~\ref{fig:intro} compares the trajectory and opinion of a robot using NOD (adapted from \cite{ref:cathcart2023proactive}) and S-NOD~\eqref{eq:ENOD} models to navigate around a human mover. The S-NOD robot passes the human with a minimum distance of 0.96m and arrives at its goal in 14.4 seconds. Without the return to the sensitive bifurcation point that S-NOD provides, the NOD robot's opinion change lags as the human makes a sharp turn and the robot experiences
a collision (moving closer than 0.3m to the human) less than 6 seconds into the simulation.

Fig.~\ref{fig:robot_trajectories} showcases three robots navigating towards a common goal and around two approaching humans using multi-agent S-NOD~\eqref{eq:MA_ENOD}. The communication networks are those in Fig.~\ref{fig:multagent} and the spiking behavior of the robots in Fig.~\ref{subfig:robot_traj_opinions} 
is similar in many ways to the spiking behavior in Fig.~\ref{subfig:multagent_agreement} and~\ref{subfig:multagent_disagreement}.
However, due to the distance dependent $|a_{ik}|$ of Fig.~\ref{subfig:robot_traj_disagreement}, robot 3 prefers the same opinion as robot 1, whereas due to the constant $|a_{ik}|$ of Fig.~\ref{subfig:multagent_disagreement}, agent 3 prefers the opinion of agent 2. 
Notably in Fig.~\ref{subfig:robot_traj_disagreement}, robot 2 switches from an initial right turn (opposite to robots 1 and 3) to avoid human 1 to a left turn to avoid human 2. S-NOD gracefully navigates robot 2 out of consecutive potential collisions, implementing different turning opinions. This embodies the agility and sequential decision-making feature of S-NOD.
\section{Final Remarks}
We presented and analyzed Spiking Nonlinear Opinion Dynamics (S-NOD) for a single agent and two options. We showcased the ability of S-NOD to swiftly form opinions and regulate back to ultrasensitivity. S-NOD provides first-of-its-kind \textit{two-dimensional} excitable (spiking) dynamics for agile decision-making over two options. 
We showed how NOD can become too robust, but the self-regulation of S-NOD 
recovers flexibility. 
We analyzed existence of limit cycles for certain parameter regimes in S-NOD. 
We presented S-NOD for multiple agents that communicate opinions over a network 
and highlighted 
potential for agent (anti)synchronization. 
We illustrated 
S-NOD's agility in a social robot navigation application and plan to implement on physical robots. We aim to provide analytical guarantees on the onset of periodic spiking in limit cycles, to analyze synchronization patterns for multiple agents, and to generalize to multiple options. 


\bibliographystyle{IEEEtran}
\bibliography{references}

\begin{thebibliography}{10}
\providecommand{\url}[1]{#1}
\csname url@rmstyle\endcsname
\providecommand{\newblock}{\relax}
\providecommand{\bibinfo}[2]{#2}
\providecommand\BIBentrySTDinterwordspacing{\spaceskip=0pt\relax}
\providecommand\BIBentryALTinterwordstretchfactor{4}
\providecommand\BIBentryALTinterwordspacing{\spaceskip=\fontdimen2\font plus
\BIBentryALTinterwordstretchfactor\fontdimen3\font minus \fontdimen4\font\relax}
\providecommand\BIBforeignlanguage[2]{{%
\expandafter\ifx\csname l@#1\endcsname\relax
\typeout{** WARNING: IEEEtran.bst: No hyphenation pattern has been}%
\typeout{** loaded for the language `#1'. Using the pattern for}%
\typeout{** the default language instead.}%
\else
\language=\csname l@#1\endcsname
\fi
#2}}

\bibitem{dayan2005theoretical}
P.~Dayan and L.~F. Abbott, \emph{Theoretical Neuroscience: Computational and Mathematical Modeling of Neural Systems}.\hskip 1em plus 0.5em minus 0.4em\relax MIT press, 2005.

\bibitem{Balazsi2011}
G.~Bal{\'a}zsi, A.~van Oudenaarden, and J.~J. Collins, ``Cellular decision making and biological noise: from microbes to mammals,'' \emph{Cell}, vol. 144, no.~6, pp. 910--925, 2011.

\bibitem{suel2006excitable}
G.~M. S{\"u}el, J.~Garcia-Ojalvo, L.~M. Liberman, and M.~B. Elowitz, ``An excitable gene regulatory circuit induces transient cellular differentiation,'' \emph{Nature}, vol. 440, no. 7083, pp. 545--550, 2006.

\bibitem{fromm2007electrical}
J.~Fromm and S.~Lautner, ``Electrical signals and their physiological significance in plants,'' \emph{Plant, Cell Environ.}, vol.~30, no.~3, pp. 249--257, 2007.

\bibitem{sepulchre2022spiking}
R.~Sepulchre, ``Spiking control systems,'' \emph{Proc. IEEE}, vol. 110, no.~5, pp. 577--589, 2022.

\bibitem{bartolozzi2022embodied}
C.~Bartolozzi, G.~Indiveri, and E.~Donati, ``Embodied neuromorphic intelligence,'' \emph{Nature Commun.}, vol.~13, no.~1, p. 1024, 2022.

\bibitem{franci2019sensitivity}
A.~Franci, G.~Drion, and R.~Sepulchre, ``The sensitivity function of excitable feedback systems,'' in \emph{IEEE Conf. Decis. Control (CDC)}, 2019, pp. 4723--4728.

\bibitem{ref:Sepulchre_ExcitableBehaviorsChapter2018}
R.~\vspace{0mm} Sepulchre, G.~Drion, and A.~Franci, ``Excitable behaviors,'' in \emph{Emerging Applications of Control and Systems Theory: A Festschrift in Honor of Mathukumalli Vidyasagar}, R.~Tempo, S.~Yurkovich, and P.~Misra, Eds.\hskip 1em plus 0.5em minus 0.4em\relax Springer International Publishing, 2018.

\bibitem{Priebe2016}
N.~J. Priebe, ``Mechanisms of orientation selectivity in the primary visual cortex,'' \emph{Annu. Rev. Vis. Sci.}, vol.~2, pp. 85--107, 2016.

\bibitem{gallivan2018decision}
J.~P. Gallivan, C.~S. Chapman, D.~M. Wolpert, and J.~R. Flanagan, ``Decision-making in sensorimotor control,'' \emph{Nature Reviews Neuroscience}, vol.~19, no.~9, pp. 519--534, 2018.

\bibitem{ref:Bizyaeva_NODTunable2023}
A.~Bizyaeva, A.~Franci, and N.~E. Leonard, ``Nonlinear opinion dynamics with tunable sensitivity,'' \emph{IEEE Trans. Autom. Control}, vol.~68, no.~3, pp. 1415--1430, 2023.

\bibitem{ref:Leonard_ARpaper2024}
N.~E. Leonard, A.~Bizyaeva, and A.~Franci, ``Fast and flexible multiagent decision-making,'' \emph{Annu. Rev. Control. Robot. Auton. Syst.}, vol.~7, 2024.

\bibitem{VS-NEL:16c}
V.~Srivastava and N.~E. Leonard, ``Bio-inspired decision-making and control: From honeybees and neurons to network design,'' in \emph{Am. Control Conf. (ACC)}, Seattle, WA, May 2017, pp. 2026--2039.

\bibitem{gray2018multiagent}
R.~Gray, A.~Franci, V.~Srivastava, and N.~E. Leonard, ``Multiagent decision-making dynamics inspired by honeybees,'' \emph{IEEE Trans. Control Netw. Syst.}, vol.~5, no.~2, pp. 793--806, 2018.

\bibitem{ref:franci2021analysis}
A.~Franci, A.~Bizyaeva, S.~Park, and N.~E. Leonard, ``Analysis and control of agreement and disagreement opinion cascades,'' \emph{Swarm Intell.}, vol.~15, no.~1, pp. 47--82, 2021.

\bibitem{bizyaeva2023multi}
A.~Bizyaeva, A.~Franci, and N.~E. Leonard, ``Multi-topic belief formation through bifurcations over signed social networks,'' \emph{arXiv:2308.02755}, 2023.

\bibitem{sepulchre2019control}
R.~Sepulchre, G.~Drion, and A.~Franci, ``Control across scales by positive and negative feedback,'' \emph{Annu. Rev. Control Robot. Auton. Syst.}, vol.~2, no.~1, pp. 89--113, 2019.

\bibitem{ref:cathcart2023proactive}
C.~Cathcart, M.~Santos, S.~Park, and N.~E. Leonard, ``Proactive opinion-driven robot navigation around human movers,'' in \emph{IEEE/RSJ Int. Conf. Intell. Robot. Syst. (IROS)}, 2023, pp. 4052--4058.

\bibitem{ref:golubitsky2012singularities1}
M.~Golubitsky and D.~G. Schaeffer, \emph{Singularities and Groups in Bifurcation Theory: Volume I}.\hskip 1em plus 0.5em minus 0.4em\relax Springer, 1985, vol.~51.

\bibitem{ref:strogatz2000}
S.~H. Strogatz, \emph{Nonlinear Dynamics and Chaos: With Applications to Physics, Biology, Chemistry and Engineering}.\hskip 1em plus 0.5em minus 0.4em\relax Westview Press, 2000.

\bibitem{ref:amorim2024threshold}
G.~Amorim, M.~Santos, S.~Park, A.~Franci, and N.~E. Leonard, ``Threshold decision-making dynamics adaptive to physical constraints and changing environment,'' in \emph{Eur. Cont. Conf.}, 2024, pp. 1908--1913.

\bibitem{ref:Ozcimder_ACCdance2016_2investigating}
K.~Özcimder, B.~Dey, R.~J. Lazier, D.~Trueman, and N.~E. Leonard, ``Investigating group behavior in dance: an evolutionary dynamics approach,'' in \emph{Am. Control Conf. (ACC)}, 2016, pp. 6465--6470.

\bibitem{krupa2001relaxation}
M.~Krupa and P.~Szmolyan, ``Relaxation oscillation and canard explosion,'' \emph{J. Diff. Equ.}, vol. 174, no.~2, pp. 312--368, 2001.

\bibitem{ref:arbelaiz2024excitablecrawling}
J.~Arbelaiz, A.~Franci, N.~E. Leonard, R.~Sepulchre, and B.~Bamieh, ``Excitable crawling,'' \emph{Symp. Math. Th. Net. Sys. (MTNS), arXiv: 2405.20593}, 2024.

\end{thebibliography}

\end{document}